\documentclass[11pt]{article}
\usepackage{amsmath,amssymb,amsthm}
\usepackage{thm-restate}
\usepackage{fullpage}
\usepackage{microtype}
\usepackage{needspace}
\usepackage[colorlinks=true,linkcolor=blue,citecolor=blue,urlcolor=blue]{hyperref}
\usepackage{cleveref}

\newtheorem{theorem}{Theorem}[section]
\newtheorem{lemma}[theorem]{Lemma}

\newtheorem{claim}{Claim}[section]
\theoremstyle{definition}
\newtheorem{definition}[theorem]{Definition}
\numberwithin{equation}{section}

\crefname{theorem}{theorem}{theorems}
\crefname{lemma}{lemma}{lemmas}
\crefname{corollary}{corollary}{corollaries}
\crefname{claim}{claim}{claims}
\crefname{definition}{definition}{definitions}

\newcommand{\F}{\mathbb F}
\newcommand{\OCT}{\textnormal{\textsc{Odd Cycle Transversal}}}

\title{A Deterministic Polynomial Kernel for Odd Cycle Transversal}
\author{%
Tomohiro Koana\thanks{Graduate School of Information Science and Technology,
The University of Tokyo, Japan.\newline
Email: \texttt{tomohiro.koana@gmail.com}.}
\and
Soh Kumabe\thanks{CyberAgent, Inc., Tokyo, Japan.
Email: \texttt{kumabe\_soh@cyberagent.co.jp}.}}
\date{}

\begin{document}
\maketitle

\begin{abstract}
We give a deterministic polynomial kernel for \textsc{Odd Cycle Transversal},
derandomizing the randomized kernel of Kratsch and Wahlstr\"om (TALG 2014).
Our algorithm uses a deterministic polynomial-time construction of
almost multilinear representations of gammoids. Such a representation
assigns a block of columns to each element so that, for every subset of
elements, the normalized matrix rank approximates its matroid rank to
within a prescribed additive error \(\delta\).
The construction builds on recent breakthroughs in NC algorithms for
matching. Our kernelization algorithm then computes the required
representative families from these representations.
\end{abstract}
\clearpage

\section{Introduction}\label{sec:introduction}

An \emph{odd cycle transversal} of an undirected graph \(G\) is a
set \(X\subseteq V(G)\) such that \(G-X\) is bipartite.
In \OCT{}, the input consists of \(G\) and a nonnegative integer
\(k\), and the task is to decide whether \(G\) has an odd cycle
transversal of size at most \(k\).
Reed, Smith, and Vetta established that the problem is
fixed-parameter tractable using iterative compression
\cite{ReedSmithVetta2004}.
A central question was whether \OCT{} also admits a
\emph{polynomial kernelization}: a polynomial-time reduction to an
equivalent instance of the same problem whose size is polynomial in
\(k\) \cite{CyganEtAl2015,FominEtAl2019}.
Kratsch and Wahlstr\"om resolved this question by using matroid
representations to obtain a randomized polynomial kernel
\cite{KratschWahlstrom2014}.
This work received the 2018 EATCS--IPEC Nerode Prize.
Their subsequent framework based on representative families and cut covering
extended this approach to several related problems
\cite{KratschWahlstrom2020}.
Obtaining a deterministic polynomial kernel for \OCT{} on general
graphs has remained a longstanding open problem
(see e.g., \cite{Kratsch2014,BodlaenderEtAl2016,MajumdarRamanSaurabh2018,JansenPilipczukVanLeeuwen2019,CrespelleEtAl2023,GurjarEtAl2025}).
Jansen, Pilipczuk, and van Leeuwen made partial progress toward this goal
by giving a deterministic polynomial kernel for planar graphs
\cite{JansenPilipczukVanLeeuwen2019}.
More recently, \cite{GurjarEtAl2025} obtained a deterministic reduction for
general graphs to an instance of polynomial size in quasipolynomial time.
Despite this progress, it remained open whether \OCT{} admits a
deterministic polynomial kernel on general graphs.
We resolve this question.

\begin{restatable}{theorem}{octkernel}\label{thm:oct-kernel}
\OCT{}, parameterized by the number \(k\) of deleted vertices,
admits a deterministic polynomial-time kernelization with
\(O(k^{12})\) vertices and edges.
\end{restatable}

The same algorithm, applied to the reductions of Kratsch and
Wahlstr\"om \cite{KratschWahlstrom2020}, also yields deterministic
polynomial kernels for other cut problems, including
\textsc{Almost 2-SAT}, \textsc{Group Feedback Vertex Set} over
any fixed finite group, and \textsc{Vertex Multiway Cut} with a
fixed number of undeletable terminals.

\paragraph{Deterministic cut covering}
Our key ingredient is a derandomization of the cut-covering lemma
of Kratsch and Wahlstr\"om \cite{KratschWahlstrom2020}.
Let \(D\) be a directed graph. For \(A,B\subseteq V(D)\), an
\emph{\((A,B)\)-cut} is a vertex set meeting every directed path
from \(A\) to \(B\); it may contain vertices of \(A\cup B\).
Given disjoint terminal sets \(S,T\subseteq V(D)\), the aim is to
find a small vertex set \(Z\) containing a minimum \((A,B)\)-cut
for every \(A\subseteq S\) and \(B\subseteq T\).
Kratsch and Wahlstr\"om compute such a set of size
\(O((|S|+|T|)^3)\) in randomized polynomial time.

\begin{restatable}{theorem}{cutcovering}\label{thm:cut-covering}
Given a directed graph \(D\) and disjoint terminal sets
\(S,T\subseteq V(D)\), one can compute, in deterministic polynomial
time, a set \(Z\subseteq V(D)\) of size \(O((|S|+|T|)^3)\)
containing \(S\cup T\) and a minimum \((A,B)\)-cut for every
\(A\subseteq S\) and \(B\subseteq T\).
\end{restatable}

\paragraph{Prior approach}
The cut-covering framework of Kratsch and Wahlstr\"om
\cite{KratschWahlstrom2020} identifies vertices that can be safely
bypassed using representative families in gammoids.
A \emph{gammoid} is a matroid whose independent sets are the vertex
subsets reachable from a fixed source set by vertex-disjoint directed
paths.
For a family \(\mathcal F\) of subsets of a matroid's ground set,
a \emph{representative family} is a subfamily that preserves all
independent extensions. That is, for every independent set \(Y\),
if some member of \(\mathcal F\) is disjoint from \(Y\) and its
union with \(Y\) is independent, then the subfamily also contains
such a member.
See \Cref{sec:preliminaries} for the formal definitions.

To apply this notion to cut covering, their algorithm encodes each
vertex as a triple in the direct sum of one uniform matroid and two
gammoids. A representative family of these triples retains every vertex
that belongs to all minimum cuts for some choice of terminal subsets.
Thus, unselected vertices can be safely bypassed, and repeating this
process yields a cut-covering set.

The size bound comes from Lov\'asz's exterior-algebra argument
\cite{Lovasz1977}: any family of \(d\)-element independent sets
in a rank-\(r\) linear matroid has a representative family of size
at most \(\binom{r}{d}\).
This argument also yields a deterministic algorithm for computing such
a family, as developed in Marx's work on parameterized matroid problems
\cite{Marx2009}.
The algorithm takes an explicitly listed family and a linear
representation of the matroid as input. Such a representation assigns
a vector to each element so that a set is independent exactly when
the corresponding vectors are linearly independent.

It therefore remains to construct linear representations of the
uniform matroid and the two gammoids used in the cut-covering algorithm.
For uniform matroids, such a representation can be constructed
deterministically in polynomial time using a Vandermonde matrix.
For gammoids, a polynomial-time construction is known only with
randomization \cite{Marx2009}, using polynomial identity testing
\cite{Schwartz1980,Zippel1979}.
This construction is where randomization enters the cut-covering algorithm.
We explain these constructions and the representative-family computation
in \Cref{sec:randomized-approach}.

% Earlier deterministic approaches include constructions parameterized
% by matroid rank \cite{MisraEtAl2020}.
% A transversal matroid on a bipartite graph with bipartition \((V,W)\)
% has ground set \(V\), and its independent sets are the subsets of
% \(V\) that can be matched into \(W\).
% The \emph{union representations} introduced in
% \cite{LokshtanovEtAl2018} are collections of linear matroids in
% which a set is independent in the original matroid exactly when it
% is independent in at least one member of the collection.
% That work constructed such representations for transversal matroids
% in quasipolynomial time and used them to compute representative
% families. These representations were combined into a single linear
% representation and extended to gammoids in \cite{GurjarEtAl2025},
% yielding deterministic quasipolynomial-time cut covering and
% kernelization.

\paragraph{Our approach}
The preceding discussion shows that constructing linear representations
of gammoids in deterministic polynomial time would derandomize the
cut-covering lemma. We do not know how to obtain such a construction.
Instead, we construct almost multilinear representations of gammoids
and compute from them the exact representative families needed for
cut covering, with both steps running in deterministic polynomial time.

For a matroid \(M\), let \(r_M(F)\) be
the maximum size of an independent subset of \(F\).
A \(\delta\)-almost multilinear representation \cite{KuhneYashfe2025}
has an integer scale
\(t\geq1\) and assigns a subspace \(S_M(e)\) of dimension at most
\(t\) to each element, such that
\(\bigl|t^{-1}\dim\sum_{e\in F}S_M(e)-r_M(F)\bigr|\leq\delta\)
for every \(F\) (\Cref{def:almost-multilinear}).
Linear representations give the case \(t=1\) and \(\delta=0\).
Allowing a small rank error lets us construct such representations
of gammoids in deterministic polynomial time.
We first give the construction for transversal matroids.

\begin{restatable}{theorem}{transversalrepresentation}\label{lem:transversal-blocks}
Given a transversal matroid \(T\) on a bipartite graph \(G\) and
a rational \(\delta>0\), one can deterministically construct
a prime \(p\) and a \(\delta\)-almost multilinear representation
of \(T\) over \(\F_p\). The prime \(p\) and the running time
are polynomially bounded in the input size and \(1/\delta\).
\end{restatable}

The construction builds on a recent breakthrough on bipartite matching.
Determinant computation is in NC, meaning deterministic polylogarithmic
parallel time using polynomially many processors \cite{Csanky1976}.
Applying polynomial identity testing to the determinant of the Edmonds
matrix, the symbolic biadjacency matrix of the graph, therefore gives a
randomized NC algorithm for deciding whether a perfect matching exists.
A deterministic NC algorithm remained elusive for decades, until
Chatterjee, Ghosh, Gurjar, Raj, and Thierauf proved that bipartite
matching is in NC \cite{ChatterjeeEtAl2026}.
Their method uses the explicit subspace designs of Guruswami and
Kopparty \cite{GuruswamiKopparty2016}.
Kopparty and Saraf \cite{KoppartySaraf2026} subsequently simplified
this approach by giving a criterion based on polynomial multiplicities,
with improved parameters for the resulting rank computation.
We adapt this simplified construction to obtain an almost multilinear representation.

Dualizing the construction gives the
corresponding result for gammoids.

\begin{restatable}{corollary}{gammoidrepresentation}\label{thm:almost-multilinear}
Given a gammoid \(M\) on a directed graph \(G\) with sources
\(S\subseteq V(G)\) and ground set \(V\subseteq V(G)\), and a rational
\(\delta>0\), one can deterministically construct a prime \(p\)
and a \(\delta\)-almost multilinear representation of \(M\)
over \(\F_p\). The prime \(p\) and the running time are
polynomially bounded in the input size and \(1/\delta\).
\end{restatable}

Given linear representations of \(d\) matroids of rank at most \(r\)
over a common field, Gaussian elimination in the corresponding tensor
space yields a representative family of size at most \(r^d\) for
any family of tuples containing one element from each matroid
(\Cref{subsec:linear-representative-families}).
We obtain representative families of asymptotically the same size
from sufficiently accurate almost multilinear representations.

\begin{samepage}
\begin{restatable}{theorem}{representativeselection}\label{thm:selection}
Let \(d,r,t\geq1\) be integers, let \(p\) be a prime, and let
\(M_1,\ldots,M_d\) be matroids of rank at most \(r\) on a common
ground set \(V\) of size \(n\), where \(r\leq n\).
Suppose that each \(M_i\) is supplied with a \(\delta\)-almost
multilinear representation \(A_i\) over \(\F_p\) of scale \(t\), where
\(0\leq\delta\leq1/(4d(r^{d-1}+1))\).
Given a family \(\mathcal F\subseteq V^d\), one can deterministically
compute a representative family
\(\widehat{\mathcal F}\subseteq\mathcal F\) of size at most \(2r^d\)
in time \(N^{O(d)}\), where \(N\) is the total bit length of
\(A_1,\ldots,A_d\).
\end{restatable}
\end{samepage}

The algorithm assigns a tensor subspace to each candidate and retains
it when it increases the span dimension by at least \(t^d/2\).
The ambient dimension bounds the number retained by \(2r^d\).
For every independent set that is extended by some candidate in $\mathcal F$, the first
such candidate increases the span dimension by at least \(t^d/2\)
and is therefore included in the representative family.
In the two-way cut-covering construction, one
uniform matroid and two gammoids give \(d=3\) and
\(r=\max\{|S|,|T|\}\). Thus the number of retained candidates is
\(O((|S|+|T|)^3)\), independently of the representation scale.

\paragraph{Organization}
\Cref{sec:preliminaries} collects the definitions.
\Cref{sec:randomized-approach} reviews prior methods for randomized
cut covering. Readers familiar with these methods may skip this section.
\Cref{sec:blocks} develops the polynomial tools and constructs the
almost multilinear representations,
and \Cref{sec:selection} proves \Cref{thm:selection}.
\Cref{sec:applications} derandomizes cut covering and proves the
\OCT{} kernel in \Cref{thm:oct-kernel}.

\section{Preliminaries}\label{sec:preliminaries}

\paragraph{Graphs}
All graphs are finite. We write \(V(G)\) and \(E(G)\) for the
vertex and edge sets, and \(G-X\) for the graph obtained by
deleting \(X\subseteq V(G)\).
Paths have no repeated vertices, respect arc directions in directed
graphs, and may have length zero. Vertex-disjoint paths share no
vertices, including endpoints. A matching is a set of pairwise
vertex-disjoint edges.

Let \(D\) be a directed graph and \(A,B\subseteq V(D)\).
An \emph{\((A,B)\)-cut} is a vertex set meeting every directed
path from \(A\) to \(B\); it may intersect \(A\cup B\).
By Menger's theorem, its minimum size equals the maximum number
of pairwise vertex-disjoint paths from \(A\) to \(B\).

\paragraph{Matroids}
A \emph{matroid} \(M=(V,\mathcal I)\) has a finite ground set
\(V\) and a nonempty family \(\mathcal I\subseteq2^V\) of
\emph{independent sets}, closed under taking subsets and satisfying
the exchange axiom: if \(I,J\in\mathcal I\) with \(|I|<|J|\), some
\(e\in J\setminus I\) satisfies \(I\cup\{e\}\in\mathcal I\).
A \emph{basis} is an inclusion-wise maximal independent set.
For \(F\subseteq V\), the \emph{rank} \(r_M(F)\) is the maximum
size of an independent subset of \(F\). The rank of \(M\) is \(r_M(V)\).

For a matroid \(M=(V,\mathcal I)\) and \(F\subseteq V\), the
\emph{restriction} \(M|F\) has ground set \(F\) and independent
sets \(\{I\in\mathcal I:I\subseteq F\}\).
The \emph{dual} \(M^*\) has ground set \(V\) and bases
\(V\setminus B\), where \(B\) ranges over the bases of \(M\).
The dual rank formula is
\(r_{M^*}(F)=|F|-r_M(V)+r_M(V\setminus F)\).

For matroids \(M_1,\ldots,M_d\) on pairwise disjoint ground sets
\(V_1,\ldots,V_d\), their \emph{direct sum}
\(M:=\bigoplus_{i=1}^d M_i\) has ground set
\(V:=\bigcup_{i=1}^d V_i\). A set \(I\subseteq V\) is independent
in \(M\) if and only if \(I\cap V_i\) is independent in \(M_i\)
for every \(i\).

The following are examples of matroids used in this paper.

\paragraph{Uniform matroids}
For a finite set \(V\) and an integer \(0\leq r\leq|V|\), the
\emph{uniform matroid} of rank \(r\) on \(V\) has as its independent
sets all subsets of \(V\) of size at most \(r\).

\paragraph{Transversal matroids}
Let \(G\) be a bipartite graph with bipartition \((V,W)\).
The \emph{transversal matroid} on \(G\) has ground set
\(V\). A set \(I\subseteq V\) is independent if and only if
\(G\) has a matching covering every vertex of \(I\).
In algorithmic statements, \(G\) and its bipartition are part of the input.

\paragraph{Gammoids}
Let \(G\) be a directed graph and \(S\subseteq V(G)\).
The \emph{strict gammoid} \(M\) on \(G\) with sources \(S\) has
ground set \(V(G)\) \cite{Mason1972}.
A set \(I\subseteq V(G)\) is independent if and only if there
are \(|I|\) pairwise vertex-disjoint directed paths from \(S\)
whose end vertices are exactly \(I\). Paths of length zero are allowed.
A \emph{gammoid} on \(G\) is a restriction \(M|V\) to a set
\(V\subseteq V(G)\). The restriction keeps the same independence
condition for \(I\subseteq V\), with paths still taken in \(G\).
In algorithmic statements, \(G,S,V\) are part of the input.

\paragraph{Linear representations}
A \emph{linear representation} of a matroid \(M\) on \(V\) over
a field \(\F\) is a matrix over \(\F\) with columns indexed by \(V\),
such that a set is independent in \(M\) if and only if its columns
are linearly independent. A matroid admitting such a representation
is \emph{linear} over \(\F\).
The three classes of matroids above are all linear.
For uniform matroids, a linear representation can be constructed
deterministically in polynomial time using a Vandermonde matrix.
For transversal matroids and gammoids, randomized polynomial-time
constructions are known, while deterministic polynomial-time
constructions remain open \cite{Marx2009,GurjarEtAl2025}.
We describe the randomized constructions in \Cref{sec:randomized-approach}.

\paragraph{Representative families}
Let \(M\) be a matroid on \(V\). A set \(A\subseteq V\)
\emph{extends} a set \(Y\subseteq V\) if
\(A\cap Y=\varnothing\) and \(A\cup Y\) is independent in \(M\).
Let \(\mathcal F\) be a family of subsets of \(V\).
A subfamily \(\widehat{\mathcal F}\subseteq\mathcal F\)
is a \emph{representative family} for \(\mathcal F\) if, for every
\(Y\subseteq V\), whenever some member of \(\mathcal F\) extends
\(Y\), some member of \(\widehat{\mathcal F}\) extends \(Y\).

For matroids \(M_1,\ldots,M_d\) on a common ground set \(V\),
we identify each tuple \((v_1,\ldots,v_d)\in V^d\) with the set
\(\{(v_i,i):1\leq i\leq d\}\) in their direct sum on disjoint
copies \(V\times\{i\}\).
Representativeness for a family \(\mathcal F\subseteq V^d\) refers
to this direct sum.

\paragraph{Linear algebra}
All vector spaces used below are finite-dimensional. For subspaces
\(U_1,\ldots,U_d\) of a common vector space, their sum is the span
of their union. For \(U_i\subseteq \mathcal V_i\) over a common field,
the tensor subspace
\(\bigotimes_{i=1}^d U_i\subseteq\bigotimes_{i=1}^d \mathcal V_i\)
is spanned by the tensors \(u_1\otimes\cdots\otimes u_d\)
with \(u_i\in U_i\), and has dimension \(\prod_i\dim U_i\).
If \(\pi:\mathcal V\to\mathcal V/W\) is the quotient map and \(U\subseteq\mathcal V\)
is a subspace, then \(\dim\pi(U)=\dim(U+W)-\dim W\).

\section{Prior methods for randomized cut covering}\label{sec:randomized-approach}

The randomized cut-covering algorithm of Kratsch and Wahlstr\"om
\cite{KratschWahlstrom2020} uses two algebraic steps: constructing
linear representations of gammoids and computing a representative family
in their direct sum with a uniform matroid.
We first describe randomized linear representations of transversal
matroids and gammoids in \Cref{subsec:randomized-representations},
then explain how to compute representative families for tuples
containing one element from each summand in
\Cref{subsec:linear-representative-families}.
The latter step is deterministic once the representations are given.
Finally, \Cref{subsec:randomized-cut-covering} explains how these
families identify vertices that can be safely bypassed in the
cut-covering construction.

\subsection{Randomized linear representations of transversal matroids and gammoids}
\label{subsec:randomized-representations}

A transversal matroid can be represented by the Edmonds matrix of
its bipartite graph.
Let \(T\) be the transversal matroid of a bipartite graph \(G\)
with bipartition \((V,W)\), and let \(n:=|V|\geq1\).
Fix a prime \(p\) and introduce one indeterminate
\(x_{vw}\) for each edge \(vw\in E(G)\).
Over \(\F_p\), the Edmonds matrix \(A(\mathbf x)\) has rows
indexed by \(W\), columns indexed by \(V\), and entries
\[
A(\mathbf x)_{w,v}:=
\begin{cases}
x_{vw},&\text{if }vw\in E(G),\\
0,&\text{otherwise.}
\end{cases}
\]
For \(I\subseteq V\) and \(J\subseteq W\) with \(|I|=|J|\),
the determinant of the submatrix on rows \(J\) and columns \(I\)
is a signed sum over the perfect matchings between \(I\) and \(J\).
Distinct matchings use distinct sets of edge variables and hence give
distinct monomials. Thus this determinant is a nonzero polynomial
exactly when such a matching exists.
It follows that \(A(\mathbf x)\) represents \(T\) over the rational
function field \(\F_p(\mathbf x)\).

To obtain a matrix over \(\F_p\), substitute an independently and
uniformly chosen field element for each variable.
A dependent set of \(T\) remains dependent under every substitution,
since all of its relevant minors are identically zero.
For each nonempty independent set \(I\), fix a matching covering \(I\)
and let \(J\) be its matched vertices in \(W\).
The corresponding determinant is a nonzero polynomial of degree
\(|I|\leq n\), so the Schwartz--Zippel bound
\cite{Schwartz1980,Zippel1979} gives probability at most \(n/p\)
that it vanishes. By a union bound over at most \(2^n\) independent
sets, the probability that the resulting matrix fails to represent
\(T\) is at most \(n2^n/p\).
Choosing a prime
\(2n2^n/\varepsilon<p<4n2^n/\varepsilon\),
the overall success probability is at least \(1-\varepsilon\).
Each field element can be encoded using
\(\lceil\log_2 p\rceil=O(n+\log(1/\varepsilon))\) bits.

We obtain representations of gammoids by dualizing.
Let \(G\) be a directed graph, let \(S\subseteq V(G)\),
and let \(M\) be the strict gammoid on \(G\) with sources \(S\).
Create a bipartite graph \(G_T\) with left side \(V(G)\) and right side
a disjoint copy of \(V(G)\setminus S\).
The right copy of \(u\) is adjacent to left vertex \(u\) and to
every left vertex \(v\) for which \(v\to u\) is an arc of \(G\).
If \(T\) is the transversal matroid of \(G_T\), then
\(M=T^*\) by transversal--gammoid duality \cite{IngletonPiff1973}.

Given a linear representation \(A\) of a matroid \(N\) on \(V\),
we construct a representation of \(N^*\) over the same field.
By Gaussian elimination and reordering columns, we obtain the form
\(A\) below and its dual representation \(H\), where \(r=\operatorname{rank}A\):
\[
A=\begin{bmatrix}I_r&D\end{bmatrix},
\qquad
H=\begin{bmatrix}-D^\top&I_{|V|-r}\end{bmatrix}.
\]

The row space of \(H\) is \(\ker A\), of dimension \(|V|-r\).
For \(F\subseteq V\), write \(A[F]\) and \(H[F]\) for the
submatrices on columns \(F\).
The row space of \(H[F]\) is obtained by projecting \(\ker A\)
onto \(F\). The vectors sent to zero are exactly those in \(\ker A\)
whose coordinates in \(F\) vanish. Their remaining coordinates form
\(\ker A[V\setminus F]\). Rank-nullity therefore gives
\[
\begin{aligned}
\operatorname{rank}H[F]
&=(|V|-r)-\dim\ker A[V\setminus F]\\
&=(|V|-r)-\bigl(|V\setminus F|-\operatorname{rank}A[V\setminus F]\bigr)\\
&=|F|-r_N(V)+r_N(V\setminus F)
=r_{N^*}(F).
\end{aligned}
\]
Thus \(H\) represents \(N^*\).

Applying this algorithm to a randomized representation of \(T\)
gives a representation of \(M=T^*\).

\subsection{Representative families in direct sums}
\label{subsec:linear-representative-families}

The classical method for computing representative families is based
on exterior algebra \cite{Lovasz1977,Marx2009,FominEtAl2016}.
For the direct sums considered here, each candidate contains one element
from each summand, so we can use simpler tensor algebra instead.
Let \(M_1,\ldots,M_d\) be matroids of rank at most \(r\)
on a common ground set \(V\), with linear representations
\(A_i\in\F^{r\times |V|}\) over the same field.
Write \(a_i(v)\) for the column of \(A_i\) indexed by \(v\).
Their direct sum on disjoint copies of \(V\) is represented by
the block diagonal matrix \(\operatorname{diag}(A_1,\ldots,A_d)\).

Let \(\mathcal F\subseteq V^d\) be an explicitly listed family.
For each tuple \(F=(v_1,\ldots,v_d)\in\mathcal F\), form the vector
\[
\Phi(F):=a_1(v_1)\otimes\cdots\otimes a_d(v_d)
\in (\F^r)^{\otimes d}\cong\F^{r^d}.
\]
Its coordinate indexed by \((j_1,\ldots,j_d)\in\{1,\ldots,r\}^d\)
is \(\prod_{i=1}^d a_i(v_i)_{j_i}\).
These are exactly the possibly nonzero \(d\times d\) minors of
the columns corresponding to \(F\) in the block diagonal matrix:
a nonzero minor must select one row from each of the $d$ row blocks.

Form the matrix whose columns are the vectors \(\Phi(F)\), and
use Gaussian elimination to select a basis of its column space.
Retain the corresponding tuples as \(\widehat{\mathcal F}\).
Since the tensor space has dimension \(r^d\), we have
\(
|\widehat{\mathcal F}|\leq r^d.
\)
Computing a column basis
 takes \(O(|\mathcal F|\cdot r^{2d})\) field operations
 by Gaussian elimination.

To prove that \(\widehat{\mathcal F}\) is a representative family for
\(\mathcal F\), fix independent sets \(Y_i\) in \(M_i\), and let
\[
L_i:=\operatorname{span}\{a_i(y):y\in Y_i\},\qquad
\pi_i:\F^r\longrightarrow\F^r/L_i.
\]
Here \(\pi_i\) is the quotient map.
For any \(v\in V\), the vector \(\pi_i(a_i(v))\) is nonzero
exactly when \(a_i(v)\notin L_i\), which is equivalent to
\(\{v\}\) extending \(Y_i\) in \(M_i\).
The tensor map \(\pi:=\bigotimes_{i=1}^d\pi_i\) satisfies
\[
\pi(\Phi(F))
=\pi_1(a_1(v_1))\otimes\cdots\otimes\pi_d(a_d(v_d)).
\]
This tensor is nonzero exactly when \(\pi_i(a_i(v_i))\neq0\) for every \(i\).
Thus \(\pi(\Phi(F))\neq0\) exactly when \(F\) extends the set
given by \(Y_1,\ldots,Y_d\) in the direct sum.
Suppose some \(F\in\mathcal F\) extends this set.
Since \(\Phi(F)\) is a linear combination of the vectors \(\Phi(F')\)
with \(F'\in\widehat{\mathcal F}\), its nonzero image under \(\pi\)
implies that \(\pi(\Phi(F'))\neq0\) for some \(F'\in\widehat{\mathcal F}\).
This tuple \(F'\) extends the same set,
proving that \(\widehat{\mathcal F}\) is representative.

\subsection{From representative families to cut covering}
\label{subsec:randomized-cut-covering}

We describe the cut-covering argument of Kratsch and Wahlstr\"om
\cite{KratschWahlstrom2020}, which is also presented in the
kernelization textbook \cite{FominEtAl2019}.
Let \(D\) be a directed graph with nonempty disjoint terminal sets
\(S,T\), where the vertices in \(S\) have no incoming arcs and those
in \(T\) have no outgoing arcs, and let \(r:=\max\{|S|,|T|\}\).
To enforce the arc conditions, add a fresh source \(s^+\) with an arc
\(s^+\to s\) for each \(s\in S\) and a fresh sink \(t^-\) with an arc
\(t\to t^-\) for each \(t\in T\), and use the copies as terminals.
This preserves minimum cut values, and mapping each copy back to its
original vertex does not increase the size of a cut-covering set.

We seek a small vertex set containing a minimum \((A,B)\)-cut for
every \(A\subseteq S\) and \(B\subseteq T\).
A nonterminal vertex is \emph{essential} for \((A,B)\)-cuts if it
belongs to every minimum \((A,B)\)-cut.

The connection to gammoids uses \emph{closest sets}.
A set \(X\) is closest to a source set \(S\) if it is the unique
minimum \((S,X)\)-cut.
Add a sink copy \(x'\) of each \(x\in X\), with the same
in-neighbors and no outgoing arcs, and let \(M\) be the gammoid
of the resulting graph with sources \(S\).
The key observation of \cite{KratschWahlstrom2020} is that \(X\) is
closest to \(S\) exactly when \(X\cup\{x'\}\) is independent in \(M\) for every
\(x\in X\setminus S\).
Such an independent set corresponds to paths ending at all vertices
of \(X\), with two paths ending at \(x\), disjoint except for this
shared endpoint. A cut avoiding \(x\) needs at least \(|X|+1\)
vertices to meet all these paths. This explains why the extra sink
copy detects whether a minimum cut can avoid \(x\).

Construct \(D_1\) by adding a sink copy of every vertex of \(D\),
and construct \(D_2\) by first reversing all arcs and then adding
sink copies. Let \(M_1\) and \(M_2\) be the gammoids on these
graphs with sources \(S\) and \(T\), respectively, and let \(M_0\)
be the rank-\(r\) uniform matroid on \(V(D)\).
For each \(v\in V(D)\setminus(S\cup T)\), define
\(F_v:=(v_0,v'_1,v'_2)\) in the direct sum
\(M_0\oplus M_1\oplus M_2\), where \(v_0\) is the copy of \(v\)
in \(M_0\) and \(v'_i\) is its sink copy in \(M_i\).
Let \(\mathcal F:=\{F_v:v\in V(D)\setminus(S\cup T)\}\) be the
family of all such triples.
All three matroids have rank at most \(r\), so the tensor construction
gives a representative family
\(\widehat{\mathcal F}\subseteq\mathcal F\) for \(\mathcal F\)
of size at most \(r^3\).

For \(A\subseteq S\) and \(B\subseteq T\), let \(C_A\) and \(C_B\)
be the minimum \((A,B)\)-cuts closest to \(A\) in \(D\) and to \(B\)
in the reversed graph, respectively.
A nonterminal \(u\) is essential for \((A,B)\)-cuts exactly when
\(u\in C_A\cap C_B\) \cite{KratschWahlstrom2020}.
Fix an essential vertex \(v\) and set
\[
Y_0:=C_A\setminus\{v\},\qquad
Y_1:=C_A\cup(S\setminus A),\qquad
Y_2:=C_B\cup(T\setminus B).
\]
Using the key observation,  a nonterminal
\(u\) is essential exactly when \(u'_1,u'_2\) extend \(Y_1,Y_2\), respectively.
Since every essential \(u\) belongs to \(Y_0\) except $v$,
\(F_v\) is the unique triple extending \(Y_0,Y_1,Y_2\).
Consequently,
\(F_v\in\widehat{\mathcal F}\).

If \(F_v\notin\widehat{\mathcal F}\), then \(v\) is \emph{irrelevant}:
every pair \(A,B\) has a minimum cut avoiding \(v\).
We may therefore \emph{bypass} \(v\): add an arc \(a\to b\)
for each two-arc path \(a\to v\to b\), then delete \(v\).
This operation preserves every minimum cut value.
We repeat
until
at most \(r^3+|S|+|T|\) vertices remain.

\section{Almost multilinear representations}\label{sec:blocks}

As described in \Cref{sec:randomized-approach}, linear representations
of transversal matroids and gammoids can be constructed in randomized
polynomial time, but no deterministic polynomial-time construction is known.
In this section, we deterministically construct \emph{almost multilinear
representations} for both classes. These assign a block of columns to
each element and approximate matroid ranks after normalization.
This notion was recently introduced by K\"uhne and
Yashfe~\cite{KuhneYashfe2025}.

\begin{definition}[{\cite[Definition~2.4]{KuhneYashfe2025}}]\label{def:almost-multilinear}
Let \(M\) be a matroid on \(V\), and let \(\delta\geq0\).
A \emph{\(\delta\)-almost multilinear representation} of \(M\)
over a field \(\F\) consists of an integer \(t\), called
its \emph{scale}, and a matrix \(A\in\F^{m\times t|V|}\)
for some integer \(m\geq1\). The columns of \(A\) are partitioned
into blocks \(A[e]\) of \(t\) columns indexed by \(e\in V\).
For \(F\subseteq V\), let \(A[F]\) be the concatenation of the
blocks indexed by \(F\). We require that
\begin{equation}\label{eq:almost-multilinear}
\left|\frac1t\operatorname{rank}A[F]-r_M(F)\right|\leq\delta
\qquad\text{for every }F\subseteq V.
\end{equation}
\end{definition}

For any rational \(\delta>0\), our constructions run in time polynomial
in the input size and \(1/\delta\).
We first establish the required polynomial lemmas in
\Cref{subsec:polynomial-multiplicities}, then construct representations
of transversal matroids in \Cref{subsec:transversal-matroids}.
Finally, we obtain representations of gammoids by dualizing and
restricting in \Cref{subsec:gammoids}.

\subsection{Polynomial independence and multiplicity bounds}\label{subsec:polynomial-multiplicities}

The following lemmas use the Wronskian method for explicit subspace
designs \cite{GuruswamiKopparty2016} in the form developed by
Kopparty and Saraf \cite{KoppartySaraf2026}.
For background on the Wronskian criterion for linear independence,
see Bostan and Dumas \cite{BostanDumas2010}.

Let \(p\) be a prime and \(\F=\F_p\).
For \(g\in\F[Z]\), write \([Z^h]g(Z)\) for the coefficient
of \(Z^h\) in \(g(Z)\).
For a nonzero polynomial \(f\in\F[X]\) and \(\alpha\in\F\),
let \(\operatorname{mult}(f,\alpha)\) be the largest integer \(a\)
such that \((X-\alpha)^a\) divides \(f\). Write \(f^{(h)}\)
for the \(h\)-th formal derivative of \(f\). The \emph{Wronskian}
of \(f_1,\ldots,f_s\) is the polynomial
\(\operatorname{Wr}(f_1,\ldots,f_s):=
\det(f_i^{(j-1)})_{i,j=1}^s\).

\begin{lemma}[See e.g., {\cite[Theorem~2]{BostanDumas2010}}]\label{lem:wronskian}
Let \(s\geq1\), and let \(f_1,\ldots,f_s\in\F[X]\) have
degree at most \(d<p\). Then \(f_1,\ldots,f_s\) are linearly
independent over \(\F\) if and only if
\(\operatorname{Wr}(f_1,\ldots,f_s)\ne0\).
\end{lemma}

For completeness, we also include short proofs of the following two
lemmas from \cite{KoppartySaraf2026}.

\begin{lemma}[{\cite[Lemma~2.2]{KoppartySaraf2026}}]\label{lem:multiplicity}
Let \(W\subseteq\F[X]\) be an \(s\)-dimensional space of
polynomials of degree at most \(d<p\), where \(s\geq1\).
For \(\alpha\in\F\), let
\(\nu_W(\alpha):=\max_{f\in W\setminus\{0\}}
\operatorname{mult}(f,\alpha)\).
Then, for every finite set \(\Gamma\subseteq\F\),
\begin{equation}\label{eq:multiplicity}
\sum_{\alpha\in\Gamma}
\max\{\nu_W(\alpha)-s+1,0\}\leq sd.
\end{equation}
\end{lemma}

\begin{proof}
Choose a basis \(f_1,\ldots,f_s\) of the polynomial space \(W\),
and let \(R:=\operatorname{Wr}(f_1,\ldots,f_s)\).
By \Cref{lem:wronskian}, \(R\ne0\). Each term in the determinant
defining \(R\) is a product of \(s\) polynomials of degree at most
\(d\), so \(\deg R\leq sd\).

Fix \(\alpha\in\Gamma\), and choose \(f\in W\setminus\{0\}\)
attaining \(\nu_W(\alpha)\). Extend \(f\) to a basis of \(W\).
The corresponding Wronskian is a nonzero scalar multiple of \(R\),
because a change of basis multiplies the Wronskian by the determinant
of an invertible matrix over \(\F\).
For \(0\leq h<s\), the derivative \(f^{(h)}\) is divisible by
\((X-\alpha)^{\max\{\nu_W(\alpha)-h,0\}}\).
Thus every entry in the row corresponding to \(f\) has the common
factor \((X-\alpha)^{\max\{\nu_W(\alpha)-s+1,0\}}\).
This factor divides the determinant and hence \(R\).
Since a nonzero polynomial has at most its degree many roots,
counting multiplicities, we obtain
\[
\sum_{\alpha\in\Gamma}\max\{\nu_W(\alpha)-s+1,0\}
\leq\sum_{\alpha\in\Gamma}\operatorname{mult}(R,\alpha)
\leq\deg R\leq sd.\qedhere
\]
\end{proof}

For an integer \(t\geq1\), write \(\F[X]_{<t}\) for the
polynomials of degree less than \(t\). The next lemma gives an independence condition for
polynomials with prescribed zeros.

\begin{lemma}[{\cite[Lemma~2.3]{KoppartySaraf2026}}]\label{lem:polynomial-independence}
Let \(n,t,K\geq1\) be integers satisfying
\(K>n(t+n-1)\) and \(K+t-1<p\).
Let \(\alpha_1,\ldots,\alpha_n\in\F\) be pairwise distinct.
For each \(i=1,\ldots,n\), choose a nonzero polynomial
\(P_i\in\{(X-\alpha_i)^K Q\mid Q\in\F[X]_{<t}\}\).
Then \(P_1,\ldots,P_n\) are linearly independent over \(\F\).
\end{lemma}

\begin{proof}
Suppose, for a contradiction, that the polynomials are linearly
dependent. Let \(W:=\operatorname{span}\{P_1,\ldots,P_n\}\) be the
space of all their \(\F\)-linear combinations, and let \(s:=\dim W < n\).
Since the polynomials are nonzero and dependent, we have
\(1\leq s<n\). Every polynomial in \(W\) has degree at most
\(K+t-1<p\).

Set \(\Gamma:=\{\alpha_1,\ldots,\alpha_{s+1}\}\).
For each \(i=1,\ldots,s+1\), the nonzero polynomial \(P_i\in W\)
is divisible by \((X-\alpha_i)^K\), so \(\nu_W(\alpha_i)\geq K\).
Since every polynomial in \(W\) has degree at most
\(K+t-1<p\), applying \Cref{lem:multiplicity} to \(W\) at these \(s+1\) points gives
\[
(s+1)(K-s+1)
\leq s(K+t-1).
\]
We thus obtain
\(K\leq s(t+s-1)-1\), contradicting the choice of \(K\).
\end{proof}

\subsection{Transversal matroids}\label{subsec:transversal-matroids}

We adapt the proof of \cite[Theorem~1.1]{KoppartySaraf2026} to
obtain an almost multilinear representation.

\transversalrepresentation*
\begin{proof}
Let \(T\) be a transversal matroid on a bipartite graph \(G\)
with bipartition \((V,W)\) and ground set \(V\). Let
\(n:=|V|+|W|\).

\paragraph{Construction}
Let
\(t:=\lceil2n^2/\delta\rceil\), and let
\(\ell:=t+2n\) and \(K:=n(t+n)+1\).
Let $p$ be a prime with
\(K+t-1<p<2(K+t-1)\).
We assign the distinct elements
\(1,\ldots,n\in\F_p\) as evaluation points
\(\alpha_v\), \(v\in V\), and \(\beta_w\), \(w\in W\).
For each \(v\in V\), define the linear space
\[
U_v:=\{(X-\alpha_v)^K Q\mid Q\in\F_p[X]_{<t}\}.
\]
The polynomials
\((X-\alpha_v)^K,X(X-\alpha_v)^K,\ldots,
X^{t-1}(X-\alpha_v)^K\) form a basis of \(U_v\).

We construct the matrix blockwise. For each edge \(vw\in E(G)\),
form an \(\ell\times t\) block whose columns correspond to the
above basis polynomials of \(U_v\). To obtain column \(a\),
substitute \(X=\beta_w+Z\) in \(X^a(X-\alpha_v)^K\) and
expand in \(Z\).
The coefficient of \(Z^h\) becomes the entry in row \(h\),
for \(0\leq h<\ell\). For \(vw\notin E(G)\), use the zero
block. Formally, define a matrix
\(C\in\F_p^{\ell|W|\times t|V|}\)
with rows \((w,h)\), \(w\in W\) and
\(0\leq h<\ell\), and columns \((v,a)\), \(v\in V\) and
\(0\leq a<t\), by
\[
C_{(w,h),(v,a)}:=
\begin{cases}
[Z^h]\bigl((\beta_w-\alpha_v+Z)^K(\beta_w+Z)^a\bigr),
  &vw\in E(G),\\
0,&\text{otherwise}.
\end{cases}
\]
For \(v\in V\), let \(C[v]\) be the block of \(t\) columns
indexed by \(v\), and let \(S_T(v)\) be its column space.
For \(F\subseteq V\), write \(S_T(F):=\sum_{v\in F}S_T(v)\)
and \(\rho_T(F):=\dim S_T(F)\).

Clearly, the construction runs in polynomial time.

\paragraph{Correctness}
It suffices to show that
\(t\,r_T(F)\leq\rho_T(F)\leq(t+2n)r_T(F)\) for every
\(F\subseteq V\). Indeed, these bounds imply
\(0\leq\rho_T(F)/t-r_T(F)\leq2n^2/t\leq\delta\).

\Needspace{6\baselineskip}
\begin{claim}\label{clm:transversal-lower-bound}
For every \(F\subseteq V\), we have
\(\rho_T(F)\geq t\,r_T(F)\).
\end{claim}

\begin{proof}
Suppose, for a contradiction, that \(I\subseteq V\) is
inclusion-minimal among independent sets of \(T\) whose
corresponding columns in \(C\) are dependent. Choose coefficients
\(c_{v,a}\in\F_p\), \(v\in I\) and \(0\leq a<t\), not all
zero, such that for each \(w\in W\) and \(0\leq h<\ell\),
\[
\sum_{v\in I}\sum_{a=0}^{t-1}c_{v,a}C_{(w,h),(v,a)} = 0.
\]
For each \(v\in I\), define
\[
P_v(X):=(X-\alpha_v)^K\sum_{a=0}^{t-1}c_{v,a}X^a\in U_v.
\]
Each \(P_v\) is nonzero by minimality of \(I\).
Our goal is to derive a contradiction using \Cref{lem:multiplicity}.
Let \(s:=|I|\) and
\(\mathcal P:=\operatorname{span}\{P_v:v\in I\}\).
At each of the \(s\) points \(\alpha_v\), the polynomial
\(P_v\) has multiplicity at least \(K\), so
\(\nu_{\mathcal P}(\alpha_v)\geq K\).
Fix a matching covering \(I\). We will use the column dependence
to show that \(\nu_{\mathcal P}(\beta_w)\geq\ell\) at each of
its \(s\) matched right vertices \(w\).
Once this is shown,
since every polynomial in \(\mathcal P\) has degree at most \(K+t-1\),
 applying \Cref{lem:multiplicity} at these
\(2s\) distinct points gives
\[
s(K-s+1)+s(\ell-s+1)\leq s(K+t-1).
\]
This implies \(\ell\leq t+2s-3\), contradicting \(s\leq n\)
and \(\ell=t+2n\).

It remains to show that \(\nu_{\mathcal P}(\beta_w)\geq\ell\)
for each matched right vertex \(w\). Fix such a vertex and let
\(Q_w:=\sum_{v\in I:\,vw\in E(G)}P_v\in\mathcal P\).
By \Cref{lem:polynomial-independence}, the polynomials \(P_v\)
are linearly independent. Since \(w\) has a neighbor in \(I\),
we have \(Q_w\ne0\).
For each \(0\leq h<\ell\), the definition of \(C\) and the
column dependence give
\[
[Z^h]Q_w(\beta_w+Z)
=\sum_{v\in I}\sum_{a=0}^{t-1}c_{v,a}C_{(w,h),(v,a)}=0.
\]
Thus \(Z^\ell\) divides \(Q_w(\beta_w+Z)\), so
\((X-\beta_w)^\ell\) divides \(Q_w(X)\). This proves
\(\nu_{\mathcal P}(\beta_w)\geq\ell\) and completes the
contradiction. Hence the columns indexed by any independent set
of \(T\) are independent. Taking a basis of \(F\) gives
\(\rho_T(F)\geq t\,r_T(F)\).
\end{proof}

\Needspace{6\baselineskip}
\begin{claim}\label{clm:transversal-upper-bound}
For every \(F\subseteq V\), we have
\(\rho_T(F)\leq(t+2n)r_T(F)\).
\end{claim}

\begin{proof}
Take sets \(V_0\subseteq F\) and \(W_0\subseteq W\) of minimum
total size such that \(V_0\cup W_0\) meets every edge of \(G\)
incident with \(F\). By K\"onig's theorem,
\(|V_0|+|W_0|=r_T(F)\). The columns indexed by \(V_0\) contribute
at most \(t|V_0|\) to \(\rho_T(F)\). All other columns indexed by
\(F\) are supported on the \(\ell|W_0|\) rows indexed by \(W_0\).
Thus \(\rho_T(F)\leq t|V_0|+\ell|W_0|
\leq(t+2n)r_T(F)\), as required.
\end{proof}

This concludes the proof.
\end{proof}

\subsection{Gammoids}\label{subsec:gammoids}

The dual construction from \Cref{subsec:randomized-representations}
extends to almost multilinear representations.

\begin{lemma}\label{lem:dual-blocks}
Let \(A\) be a \(\delta\)-almost multilinear representation of a
matroid \(M\) on \(V\) over a field \(\F\), with scale \(t\).
Given \(A\), one can deterministically compute a
\((2\delta)\)-almost multilinear representation \(H\) of \(M^*\)
over \(\F\), with the same scale, in polynomial time.
\end{lemma}

\begin{proof}
Compute a matrix \(H\) whose rows form a basis of \(\ker A\),
with the same column blocks as \(A\).
For \(F\subseteq V\), the rank calculation in
\Cref{subsec:randomized-representations} gives
\[
\begin{aligned}
\frac1t\operatorname{rank}H[F]-r_{M^*}(F)
&=\left(\frac1t\operatorname{rank}A[V\setminus F]-r_M(V\setminus F)\right)
-\left(\frac1t\operatorname{rank}A-r_M(V)\right).
\end{aligned}
\]
Each parenthesized term has absolute value at most \(\delta\),
so the error is at most \(2\delta\).
Gaussian elimination uses a polynomial number of field operations.
\end{proof}

Restricting an almost multilinear representation to a subset of the
ground set simply discards the other column blocks, preserving its
scale and error bound.

\gammoidrepresentation*

\section{Computing representative families}\label{sec:selection}

We now show how to compute representative families from almost
multilinear representations.

\begin{samepage}
\representativeselection*
\end{samepage}

\begin{proof}
For \(i\in\{1,\ldots,d\}\) and \(e\in V\), let
\(S_i(e)\) be the column space of the block \(A_i[e]\).
For \(X\subseteq V\), let \(S_i(X):=\sum_{e\in X}S_i(e)\) and
\(\mathcal V_i:=S_i(V)\). The representation guarantee gives
\(\dim\mathcal V_i\leq t(r+\delta)\).
For each tuple \(F=(v_1,\ldots,v_d)\in\mathcal F\), form the
matrix \(B_F:=\bigotimes_{i=1}^d A_i[v_i]\).
Its column space is \(T_F:=\bigotimes_{i=1}^d S_i(v_i)\),
contained in \(\bigotimes_{i=1}^d\mathcal V_i\).

\paragraph{Algorithm}
Initialize \(\widehat{\mathcal F}:=\varnothing\) and \(L:=\{0\}\).
Process the tuples of \(\mathcal F\) in their input order.
If \(\dim(L+T_F)-\dim L\geq t^d/2\), add \(F\) to
\(\widehat{\mathcal F}\) and replace \(L\) by \(L+T_F\).
Thus \(L\) is always the sum of the subspaces of the retained tuples.

Each insertion increases its dimension by at least \(t^d/2\),
whereas \(\dim L\leq t^d(r+\delta)^d\). Consequently,
\(|\widehat{\mathcal F}|\leq2(r+\delta)^d\).
Since \(d\delta/r\leq1/8\), the binomial expansion gives
\[
(r+\delta)^d-r^d
\leq\frac{d\delta r^{d-1}}{1-d\delta/r}
\leq2d\delta r^{d-1}<\frac12.
\]
Thus \(|\widehat{\mathcal F}|<2r^d+1\), and integrality gives
\(|\widehat{\mathcal F}|\leq2r^d\).

\paragraph{Correctness}
To prove the correctness, fix independent sets \(Y_i\) in
\(M_i\) for \(i\in\{1,\ldots,d\}\).
Call a tuple \(F=(v_1,\ldots,v_d)\in\mathcal F\) \emph{good}
if \(\{v_i\}\) extends \(Y_i\) in \(M_i\) for every \(i\),
and \emph{bad} otherwise.
Since \(\widehat{\mathcal F}\subseteq\mathcal F\), any good tuple
in \(\widehat{\mathcal F}\) also belongs to \(\mathcal F\).
Thus, it suffices to prove the converse: if \(\mathcal F\) contains a good tuple,
then \(\widehat{\mathcal F}\) contains one as well.

We show that, as long as no good tuple has been retained, adding
\(T_F\) to \(L\) for any good tuple \(F\) increases \(\dim L\)
by at least \(t^d/2\). Thus the algorithm retains the first good
tuple it processes.

Let \(\pi_i:\mathcal V_i\to\mathcal V_i/S_i(Y_i)\) be the
quotient map, and let \(\pi:=\bigotimes_{i=1}^d\pi_i\).

For a good tuple \(F=(v_1,\ldots,v_d)\), adding \(v_i\) to
\(Y_i\) increases its rank in \(M_i\) by one.
Applying the representation guarantee to \(Y_i\cup\{v_i\}\)
and using that \(A_i[Y_i]\) has \(t|Y_i|\) columns gives
\[
\begin{aligned}
\dim\pi_i(S_i(v_i))
&=\dim S_i(Y_i\cup\{v_i\})-\dim S_i(Y_i)\\
&\geq t(|Y_i|+1-\delta)-t|Y_i|
=t(1-\delta).
\end{aligned}
\]
Since \(\pi(T_F)=\bigotimes_{i=1}^d\pi_i(S_i(v_i))\), we obtain
\begin{equation}\label{eq:good-image}
\dim\pi(T_F)\geq t^d(1-\delta)^d
\geq t^d(1-d\delta).
\end{equation}

For each \(i\), let \(C_i\) be the \emph{closure} of \(Y_i\)
in \(M_i\), that is,
\(C_i:=\{e\in V:r_{M_i}(Y_i\cup\{e\})=r_{M_i}(Y_i)\}\).
Since \(r_{M_i}(C_i)=r_{M_i}(Y_i)\),
\(\dim\pi_i(S_i(C_i))
=\dim S_i(C_i)-\dim S_i(Y_i)\leq2t\delta\).
A tuple \(F=(v_1,\ldots,v_d)\) is bad precisely when
\(v_i\in C_i\) for some \(i\).
For a fixed \(i\), the images of all tuples with \(v_i\in C_i\)
lie in the tensor product of \(\pi_i(S_i(C_i))\) with the other
quotient spaces. This space has dimension at most
\(2\delta t^d(r+\delta)^{d-1}\).
Summing over the \(d\) possible coordinates gives
\begin{equation}\label{eq:bad-span}
\dim\sum_{\substack{F\in\mathcal F\\F\text{ bad}}}\pi(T_F)
\leq2d\delta t^d(r+\delta)^{d-1}.
\end{equation}

Suppose that \(\mathcal F\) contains a good tuple \(F\), but every
retained tuple is bad. Let \(L_F\) be the space \(L\) immediately
before \(F\) is processed. Since \(L_F\) is spanned by subspaces
of bad tuples, \eqref{eq:bad-span} bounds \(\dim\pi(L_F)\).
A linear map cannot increase the dimension gained by adding a
subspace. Together with \eqref{eq:good-image}, this gives
\[
\begin{aligned}
\dim(L_F+T_F)-\dim L_F
&\geq\dim\pi(T_F)-\dim\pi(L_F)\\
&\geq t^d\bigl(1-d\delta-2d\delta(r+\delta)^{d-1}\bigr)\\
&\geq t^d\bigl(1-2d\delta(1+r^{d-1})\bigr)
\geq t^d/2.
\end{aligned}
\]
Here we use \((r+\delta)^{d-1}<r^{d-1}+1/2\), obtained by the
same binomial estimate as above, and the assumed bound on \(\delta\).
Thus the algorithm retains \(F\), a contradiction.
The output \(\widehat{\mathcal F}\) is therefore a representative family
for \(\mathcal F\).

It remains to bound the running time.
Each \(B_F\) has at most \(N^d\) rows and columns and can be
formed using \(N^{O(d)}\) field operations.
Maintain a basis matrix for \(L\). Gaussian elimination on its
concatenation with \(B_F\) computes \(\dim(L+T_F)-\dim L\)
and, when \(F\) is retained, a basis for \(L+T_F\), using
\(N^{O(d)}\) field operations.
Since \(|\mathcal F|\leq n^d\leq N^d\), and each field
operation takes \((\log p)^{O(1)}\) bit operations, the total
running time is \(N^{O(d)}\).
\end{proof}

\section{The Odd Cycle Transversal kernel}\label{sec:applications}

We prove \Cref{thm:oct-kernel} by implementing the cut-covering
construction of Kratsch and Wahlstr\"om \cite{KratschWahlstrom2020}
in deterministic polynomial time and combining it with a
deterministic approximation for \OCT{}.

The cut-covering construction in \Cref{subsec:randomized-cut-covering}
uses one uniform matroid and two gammoids, each of rank at most
\(r:=\max\{|S|,|T|\}\).
Choose \(\delta:=1/(12(r^2+1))\).
The construction in \Cref{subsec:gammoids} gives
\(\delta\)-almost multilinear representations of the gammoids
at a common scale over a common prime field.
Then \Cref{thm:selection} with \(d=3\) computes a representative
family of size at most \(2r^3\) in deterministic polynomial time.
Using this family in the cut-covering construction gives a set of size
\(O((|S|+|T|)^3)\), proving \Cref{thm:cut-covering}.

\cutcovering*

For undirected graphs, replace every edge by two opposite arcs.
We now use the cut-covering theorem to prove the \OCT{} kernel.

\octkernel*
\begin{proof}
Given \((G,k)\), the algorithm in \cite{KolayEtAl2020} either
correctly rejects or returns an odd cycle transversal \(X\) of size
\(O(k^2)\).
The OCT-to-cut reduction \cite{KratschWahlstrom2014}
and the terminal-copy reduction \cite{KratschWahlstrom2020}
let us apply \Cref{thm:cut-covering} with two terminal sets of size
\(O(|X|)\). Their cut correspondence gives a set \(Z\supseteq X\)
of size \(O(|X|^3)=O(k^6)\) containing a minimum odd cycle
transversal.

The remaining reduction of \cite{KratschWahlstrom2020} encodes
paths through \(G-Z\) by parity constraints on \(Z\).
After the forced-vertex reductions, represent odd-parity constraints
by edges and even-parity constraints by paths of length two.
A deleted subdivision vertex can be replaced by an endpoint without
increasing the solution size, so this gives an equivalent ordinary
\OCT{} instance with \(O(|Z|^2)=O(k^{12})\) vertices and edges
and output parameter at most \(k\).
All steps take deterministic polynomial time.
\end{proof}

\section*{Acknowledgements}
Tomohiro Koana was supported in part by JST CREST Grant Number
JPMJCR24Q2 and JST ERATO Grant Number JPMJER2301.

\section*{Declaration of generative AI use}
ChatGPT 6 Astra generated the proof of \Cref{thm:oct-kernel} and was
also used to draft the manuscript.
The authors verified and revised the proof and the manuscript and
take full responsibility for the paper.

\bibliographystyle{alphaurl}
\bibliography{references}
\end{document}